\documentclass[conference,letterpaper]{IEEEtran}

\usepackage[utf8]{inputenc} 
\usepackage[T1]{fontenc}
\usepackage{url}
\usepackage{ifthen}
\usepackage{cite}
\usepackage[cmex10]{amsmath}
\usepackage{algorithmic}
\usepackage{graphicx}
\usepackage{textcomp}
\usepackage{xcolor}
\usepackage[english]{babel}
\usepackage{tikz}
\usetikzlibrary{positioning, arrows.meta}
\usepackage{enumitem}
\usepackage{tikz-cd}
\usepackage{extarrows}
\usepackage{bm}
\usepackage{amsfonts}
\usepackage{balance}

\newtheorem{definition}{Definition}
\newtheorem{example}{Example}

\newtheorem{theorem}{Theorem}
\newtheorem{lemma}[theorem]{Lemma}

\newtheorem{remark}{Remark}

\begin{document}

\title{Universal Feature Selection with Noisy Observations and Weak Symmetry Conditions}

\author{\IEEEauthorblockN{Dier Tang}
\IEEEauthorblockA{\textit{Department of Mathematics} \\
\textit{The University of Hong Kong}\\
Hong Kong, China \\
tangde\_math@connect.hku.hk}
\and
\IEEEauthorblockN{Guangyue Han}
\IEEEauthorblockA{\textit{Department of Mathematics} \\
\textit{The University of Hong Kong}\\
Hong Kong, China \\
ghan@hku.hk}
}

\maketitle

\begin{abstract}
This paper relaxes the restrictive symmetry conditions adopted in \cite{xu2022information,huang2024universal} and extends their universal feature selection framework to accommodate noisy observations as well as attribute structures that may exhibit directional preferences. We introduce the notion of weak spherical symmetry, quantified by second-moment distances, which allows controlled deviations from rotational invariance. Under this relaxed condition, we develop a universal feature selection framework based on the singular value decomposition of the canonical dependence matrix computed from noisy data. Our main result shows that the selected features achieve asymptotically optimal error exponents up to a residual term that depends on the symmetry deviation $\delta$ and the noise levels $\eta_1, \eta_2$. When $\delta, \eta_1, \eta_2$ are relatively small, our result recovers that of \cite{huang2024universal}, thereby demonstrating that exact spherical symmetry is unnecessary. Overall, our findings highlight the robustness of the selection framework against second-moment deviations and observation noise, thereby broadening its applicability across diverse inference tasks and providing a theoretically grounded tool for universal feature selection in practical scenarios.
\end{abstract}

\begin{IEEEkeywords}
Universal feature selection, machine learning, statistical inference, weak spherical symmetry, local information geometry.
\end{IEEEkeywords}

\section{Introduction}

Feature selection is a fundamental task in machine learning and statistical inference \cite{guyon2003introduction}. Extracting relevant features retains useful information and removes redundancies, thus improving the explainability, computational efficiency, and predictive accuracy in large-scale data analysis \cite{li2024exploring,kamalov2025mathematical}.

While task-specific features can be obtained with sufficient domain knowledge, a more ambitious goal is to extract \emph{universally effective} features that remain applicable across a diverse range of tasks. Recent work by \cite[Theorem~1]{xu2022information} and \cite[Section~5.4]{huang2024universal} operates within the Markov chain framework $U \leftrightarrow X \leftrightarrow Y \leftrightarrow V$, focusing on the use of features derived from the observed variables $X$ and $Y$ to infer the latent attributes $U$ and $V$. Under the assumption that the attribute structures exhibit no directional preference---i.e., it satisfies spherical symmetry \cite{chmielewski1981elliptically, dawid1977spherical}---they establish a universal feature selection framework characterized by the singular value decomposition (SVD) of the canonical dependence matrix. The framework thus operates without any prior knowledge of latent attributes, rendering it suitable for diverse domains such as representation learning and multi-task learning \cite{bengio2013representation,caruana1997multitask}.

This paper addresses more practical scenarios where attribute structures may have directional preferences and only noisy observations $\hat{X}$ and $\hat{Y}$ are accessible. By introducing \emph{weak spherical symmetry} and demonstrating its practical significance, we develop a universal feature selection framework that extends the results of \cite[Theorem~1]{xu2022information} and \cite[Section~5.4]{huang2024universal} to noisy observations. Moreover, we show that our result recovers theirs as a special case, confirming that the proposed weak symmetry condition is sufficient for their results.

The main contributions of this paper are as follows:
\begin{itemize}
    \item A formal definition of $\delta$-spherical symmetry (Definition~\ref{def: delta spherical symmetric random matrix}) with its motivation and practical justifications.
    
    \item A universal feature selection framework with noisy observations and $\delta$-spherical symmetry (Theorem~\ref{the: main result}).
    
    \item A demonstration that our result recovers \cite[Proposition~5.10]{huang2024universal} as a special case (Remark~\ref{rem: exact is unnecessary}), illustrating the sufficiency of the proposed condition.
\end{itemize}

\section{Preliminaries}

Throughout this paper, we consider only discrete random variables with finite alphabets. For a matrix $\bm{A}$, denote by $\|\bm{A}\|_{\mathrm{F}}$ its Frobenius norm and $\bm{A}_{i,j}$ its $(i,j)$-th entry. Let $\stackrel{\text{d}}{=}$ denote equality in distribution.

We begin with some concepts in local information geometry. Let $Z$ be a random variable over alphabet $\mathcal{Z}$ with distribution $P_{Z}$. Let $\mathcal{P}^{\mathcal{Z}}$ be the space of distributions on $\mathcal{Z}$, and $\mathrm{rel}(\mathcal{P}^{\mathcal{Z}})$ be the subset of strictly positive distributions. Given $\epsilon > 0$, define the \emph{$\epsilon$-neighborhood} of $P_{Z} \in \mathrm{rel}(\mathcal{P}^{\mathcal{Z}})$ as
\begin{equation*}
    \mathcal{N}^{\epsilon}_{\mathcal{Z}} (P_{Z}) = \Big\{P \in \mathcal{P}^{\mathcal{Z}}: \sum_{z \in \mathcal{Z}} \frac{(P(z) - P_{Z}(z))^2}{P_{Z}(z)} \leq \epsilon^2 \Big\}.
\end{equation*}
A random variable $W$ is an \emph{$\epsilon$-attribute} of $Z$, if for any $w \in \mathcal{W}$, $P_{Z \mid W}(\cdot \mid w) \in \mathcal{N}^{\epsilon}_{\mathcal{Z}} (P_{Z})$. Each $\epsilon$-attribute $W$ of $Z$ is characterized by its \emph{configuration}
\begin{equation*}
    \mathcal{C}^{\epsilon}_{\mathcal{Z}}(P_{Z}) = \big\{\mathcal{W}, \{P_{W}(w), w \in \mathcal{W} \}, \{P_{Z \mid W}(\cdot \mid w), w \in \mathcal{W} \} \big\}.
\end{equation*}
Define the \emph{information matrix} $\bm{\Phi}^{Z \mid W}$ such that
\begin{equation*}
    \bm{\Phi}^{Z \mid W}_{i,j} = \frac{1}{\epsilon} \cdot \frac{P_{Z \mid W}(z_{i} \mid w_{j}) - P_{Z}(z_{i})}{\sqrt{P_{Z}(z_{i})}}.
\end{equation*}
The $j$-th column of $\bm{\Phi}^{Z \mid W}$, denoted as $\phi^{Z \mid W}_{w_{j}}$, is called the \emph{information vector} corresponding to $w_{j}$. For $k \in \mathbb{Z}^{+}$, given normalized feature functions $h^{k}=(h_1, \cdots, h_{k})^{\mathrm{T}}: \mathcal{Z} \to \mathbb{R}^{k}$ (i.e., $\mathbb{E}[h^{k}(Z)] = 0$ and $\mathbb{E}[h^{k}(Z) (h^{k}(Z))^{\mathrm{T}}] = \bm{I}$), define the corresponding \emph{feature vectors} $\psi_{i}^{Z} \leftrightarrow h_{i}$ by $\psi_{i}^{Z}(z) = \sqrt{P_{Z}(z)} \ h_{i}(z), i = 1, \dots, k$.

For a pair of random variables $X$ and $Y$ with alphabets $\mathcal{X}$ and $\mathcal{Y}$, respectively, set $K_{XY} = \min\{|\mathcal{X}|, |\mathcal{Y}|\}$. Define the \emph{canonical dependence matrix} $\tilde{\bm{B}}_{X,Y}$ as
\begin{equation*}
    (\tilde{\bm{B}}_{X,Y})_{j,i} = \frac{P_{X,Y}(x_{i}, y_{j}) - P_{X}(x_{i})P_{Y}(y_{j})}{\sqrt{P_{X}(x_{i})} \sqrt{P_{Y}(y_{j})}}.
\end{equation*}

\section{Problem Setup}

For random variables $X$, $Y$, \cite[Theorem 1]{xu2022information} and \cite[Section 5.4]{huang2024universal} consider the Markov chain $U \stackrel{(\epsilon)}{\leftrightarrow} X \leftrightarrow Y \stackrel{(\epsilon)}{\leftrightarrow} V$, where $U$ and $V$ are the $\epsilon$-attribute of $X$ and $Y$, respectively. Given i.i.d. observations $x^{N}=(x_1, \cdots, x_{N})$,   $y^{N}=(y_1, \cdots, y_{N})$ drawn from $P_{X,Y}$, the objective is to construct feature functions $f^k: \mathcal{X} \to \mathbb{R}^{k}$ and $g^k: \mathcal{Y} \to \mathbb{R}^{k}$ that enable inference of $U$ and $V$.

However, in many practical scenarios, direct access to $X$ and $Y$ is unavailable, providing only their noisy counterparts. This motivates us to consider the following generalized model,
\vspace{-0.2cm}
\begin{equation}
\label{eq: generalized Markov chain}
\begin{tikzcd}[column sep = 2em, row sep = 1.5em]
U \arrow[r, leftrightarrow, "(\epsilon)"] & X \arrow[r, leftrightarrow] \arrow[d, leftrightarrow] & Y \arrow[r, leftrightarrow, "(\epsilon)"] \arrow[d, leftrightarrow] & V \\
& \hat{X} & \hat{Y} & 
\end{tikzcd}
\end{equation}
where $\hat{X}$, $\hat{Y}$ are the noisy versions of $X$, $Y$ over the same alphabets $\mathcal{X}$, $\mathcal{Y}$, and $U$, $V$ are the $\epsilon$-attributes of $X$, $Y$, respectively, each characterized by configurations $\mathcal{C}^{\epsilon}_{\mathcal{X}}(P_{X})$ and $\mathcal{C}^{\epsilon}_{\mathcal{Y}}(P_{Y})$. Similarly as in \cite{huang2024universal}, we assume the collection of these configurations are equipped with probability measures $\mu_{U}$ and $\mu_{V}$, respectively. The noisy channels $X \to \hat{X}$ and $Y \to \hat{Y}$ are assumed to have the following transition probability matrices:
\begin{equation}
\label{eq: channel transfer matrices}
    \bm{P}_{\hat{X} \mid X} = \bm{I} + \eta_1 \bm{T}_{\hat{X} \mid X}, \quad \bm{P}_{\hat{Y} \mid Y} = \bm{I} + \eta_2 \bm{T}_{\hat{Y} \mid Y},
\end{equation}
for some $\eta_1, \eta_2 \geq 0$, where each column of $\bm{T}_{\hat{X} \mid X}$ and $\bm{T}_{\hat{Y} \mid Y}$ sums to zero. For $k \in \{1, \dots, K_{XY}-1\}$, select normalized feature functions $f^{k}: \mathcal{X} \to \mathbb{R}^{k}$ and $g^{k}: \mathcal{Y} \to \mathbb{R}^{k}$. Suppose $\hat{x}^{N}=\{\hat{x}_1, \dots, \hat{x}_{N}\}$ and $\hat{y}^{N}=\{\hat{y}_1, \dots, \hat{y}_{N}\}$ are i.i.d. samples drawn from $P_{\hat{X}, \hat{Y}}$. Define the statistics $\hat{S}^{k}$ and $\hat{T}^{k}$ as the empirical expectations of $f^{k}(\hat{X})$ and $g^{k}(\hat{Y})$, respectively. The error probabilities for inferring $U$ from $\hat{S}^{k}$, $V$ from $\hat{S}^{k}$, $U$ from $\hat{T}^{k}$, and $V$ from $\hat{T}^{k}$ are denoted by $p_{e}^{U|\hat{S}}(f^{k}, \mathcal{C}^{\epsilon}_{\mathcal{X}}(P_{X}))$, $p_{e}^{V|\hat{S}}(f^{k}, \mathcal{C}^{\epsilon}_{\mathcal{Y}}(P_{Y}))$, $p_{e}^{U|\hat{T}}(g^{k}, \mathcal{C}^{\epsilon}_{\mathcal{X}}(P_{X}))$, and $p_{e}^{V|\hat{T}}(g^{k}, \mathcal{C}^{\epsilon}_{\mathcal{Y}}(P_{Y}))$, respectively. The (average) error exponents are then defined as follows:
\begin{align*}
    & \bar{E}^{U \mid \hat{S}}(f^{k}) = \lim\limits_{N \to \infty} - \frac{\mathbb{E}_{\mu_{U}} \big[\log p_{e}^{U|\hat{S}}(f^{k}, \mathcal{C}^{\epsilon}_{\mathcal{X}}(P_{X})) \big]}{N}, \\ & \bar{E}^{V \mid \hat{S}}(f^{k}) = \lim\limits_{N \to \infty} - \frac{\mathbb{E}_{\mu_{V}} \big[\log p_{e}^{V|\hat{S}}(f^{k}, \mathcal{C}^{\epsilon}_{\mathcal{Y}}(P_{Y})) \big]}{N}, \\ & \bar{E}^{U \mid \hat{T}}(g^{k}) = \lim\limits_{N \to \infty} - \frac{\mathbb{E}_{\mu_{U}} \big[\log p_{e}^{U|\hat{T}}(g^{k}, \mathcal{C}^{\epsilon}_{\mathcal{X}}(P_{X})) \big]}{N}, \\ & \bar{E}^{V \mid \hat{T}}(g^{k}) = \lim\limits_{N \to \infty} - \frac{\mathbb{E}_{\mu_{V}} \big[\log p_{e}^{V|\hat{T}}(g^{k}, \mathcal{C}^{\epsilon}_{\mathcal{Y}}(P_{Y})) \big]}{N}.
\end{align*}
This paper aims to seek features $f^{k}$ and $g^{k}$ that maximize the above error exponents, thereby achieving universal effectiveness across a diverse range of inference tasks.

Non‑trivial probability measures $\mu_U$ and $\mu_V$ render the information matrices $\Phi^{X \mid U}$ and $\Phi^{Y \mid V}$ random. For random matrices, a notion of spherical symmetry (or rotational invariance) has been defined \cite{chmielewski1981elliptically, dawid1977spherical}, which is widely regarded as an extension of isotropy \cite[Section 3.2.2]{vershynin2020high}.

\begin{definition}[Spherical symmetry]
\label{def:Spherical symmetric random matrix}
    A random matrix $\bm{A} \in \mathbb{R}^{n \times m}$ is \emph{spherically symmetric}, if for any deterministic orthogonal matrices $\bm{Q}_1 \in \mathbb{R}^{n \times n}$ and $\bm{Q}_2 \in \mathbb{R}^{m \times m}$, we have $\bm{Q}_1^{\mathrm{T}} \bm{A} \bm{Q}_2 \stackrel{\text{d}}{=} \bm{A}$.
\end{definition}

In the absence of prior information about $U$ and $V$, \cite[Section~5.4]{huang2024universal} assumes that $\mu_U$ and $\mu_V$ have no directional preference, i.e., $\Phi^{X \mid U}$ and $\Phi^{Y \mid V}$ are spherically symmetric. Notice that setting $\eta_1 = \eta_2 = 0$ yields $\hat{X} = X$, $\hat{Y} = Y$, thereby the system \eqref{eq: generalized Markov chain} boils down to the original Markov chain where $X$, $Y$ are respectively replaced by $\hat{X}$, $\hat{Y}$, i.e.,
\begin{equation}
\label{eq: Markov chain}
\begin{tikzcd}[sep = 2em]
U \arrow[r, leftrightarrow, "(\epsilon)"] & \hat{X} \arrow[r, leftrightarrow] & \hat{Y} \arrow[r, leftrightarrow, "(\epsilon)"]  & V.
\end{tikzcd}
\end{equation}
The result of \cite[Proposition 5.10]{huang2024universal} is reformulated below in Theorem \ref{the: Huang Proposition 5.10} using our notation.

\begin{theorem}[Huang et al. 2024]
\label{the: Huang Proposition 5.10}
    Consider the Markov chain \eqref{eq: Markov chain}. For $k \in \{1, \dots, K_{XY}-1 \}$, select normalized feature functions $f^{k}: \mathcal{X} \to \mathbb{R}^{k}$ and $g^{k}: \mathcal{Y} \to \mathbb{R}^{k}$. Assume under each probability measures, both $\bm{\Phi}^{X \mid U}$ and $\bm{\Phi}^{Y \mid V}$ are spherically symmetric. Then,
    \begin{equation}
    \label{eq: Huang Proposition 5.10 equation}
    \begin{aligned}
    & \big( \bar{E}^{U \mid \hat{S}}(f^{k}), \bar{E}^{V \mid \hat{S}}(f^{k}), \bar{E}^{U \mid \hat{T}}(g^{k}), \bar{E}^{V \mid \hat{T}}(g^{k}) \big) \leq \\ & \big( C_{U} \epsilon^2 k, \ C_{V} \epsilon^2 \sum_{i=1}^{k} \sigma_{i}^2, \ C_{U} \epsilon^2 \sum_{i=1}^{k} \sigma_{i}^2, \ C_{V} \epsilon^2 k \big) + \mathrm{o}(\epsilon^2),
    \end{aligned}
    \end{equation}
    as $\epsilon \to 0$, where $C_{U}$ and $C_{V}$ are positive constants independent of $\epsilon$, $k$ or $P_{\hat{X},\hat{Y}}$, and $\sigma_{i}$ is the $i$-th largest singular value of $\tilde{\bm{B}}_{\hat{X}, \hat{Y}}$. Moreover, the equality holds when the feature vectors corresponding to $f_{i}$ and $g_{i}$ are precisely the right and left singular vectors of $\tilde{\bm{B}}_{\hat{X}, \hat{Y}}$ corresponding to $\sigma_{i}$, respectively.
\end{theorem}

\section{Universal Feature Selection Framework}

While it is natural to assume that the distributions of $\bm{\Phi}^{X \mid U}$ and $\bm{\Phi}^{Y \mid V}$ exhibit no directional preference, the following lemma shows that exact spherical symmetry is rather restrictive and therefore challenging to achieve.

\begin{lemma}
\label{lem: 3 properties of ss}
    Let $\bm{A}$ be a random matrix. If $\bm{A}$ is spherically symmetric, then,
    \begin{enumerate}[label=(\arabic*)]
        \item\label{enu: lem_3 properties_1} 0-mean: $\mathbb{E}[\bm{A}]=\bm{0}_{n \times m}$.
        
        \item\label{enu: lem_3 properties_2} Identical distribution: For any $i,j,k,l$, $\bm{A}_{i,j} \stackrel{\text{d}}{=} \bm{A}_{k,l}$.
        
        \item\label{enu: lem_3 properties_3} Uncorrelated: For $(i,j) \neq (k,l)$, $\mathrm{Cov}(\bm{A}_{i,j}, \bm{A}_{k,l}) = 0$.
    \end{enumerate}
\end{lemma}

\begin{IEEEproof}
    In Definition \ref{def:Spherical symmetric random matrix}, setting $\bm{Q}_1 = -\bm{I}_{n}$, $\bm{Q}_2 = \bm{I}_{m}$, we have $\mathbb{E}[\bm{A}]=\bm{0}_{n \times m}$. For any $i,j,k,l$, setting $\bm{Q}_1=\bm{Q}_{ik}^{\mathrm{T}}$, $\bm{Q}_2=\bm{Q}_{jl}$ to be the permutation matrices, we have $\bm{A}_{i,j} \stackrel{\text{d}}{=} \bm{A}_{k,l}$. Setting $\bm{Q}_1 = \bm{D}_{n}^{i} \in \mathbb{R}^{n \times n}$, the diagonal matrix whose $i$-th diagonal entry is $-1$ and the rest are $1$, and $\bm{Q}_2 = \bm{I}_{m}$, we have $\forall i \neq k$, $\mathbb{E}[\bm{A}_{i,j} \bm{A}_{k,l}] = 0$. Similarly, $\forall j \neq l$, $\mathbb{E}[\bm{A}_{i,j} \bm{A}_{k,l}] = 0$. We then derive $\mathrm{Cov}(\bm{A}_{i,j}, \bm{A}_{k,l}) = 0, \forall (i,j) \neq (k,l)$.
\end{IEEEproof}

Given a random variable $Z$ and $\epsilon > 0$, assume $W$ is an $\epsilon$-attribute of $Z$. Then, for any $w \in \mathcal{W}$, $P_{Z \mid W}(\cdot \mid w) \in \mathcal{N}^{\epsilon}_{\mathcal{Z}} (P_{Z})$. Expanding the log-likelihood ratio gives
\[
\log \frac{P_{Z \mid W}(z \mid w)}{P_{Z}(z)} = \epsilon \frac{\phi_{w}(z)}{\sqrt{P_{Z}(z)}} - \frac{\epsilon^2}{2} \big( \frac{\phi_{w}(z)}{\sqrt{P_{Z}(z)}} \big)^2 + \mathrm{o}(\epsilon^2),
\]
where $\phi_{w}$ denotes the information vector corresponding to $w$. For any smooth one-parameter family $\{w(\theta)\}$ passing through $w$ at $\theta = 0$, the score function \cite[Section 11.10]{cover2012elements} satisfies
\[
\frac{\partial}{\partial \theta} \log P(z \mid w(\theta)) \bigg|_{\theta = 0} = \epsilon \frac{\phi'_{w}(z)}{\sqrt{P_{Z}(z)}} + \mathrm{O}(\epsilon^2),
\]
where $\phi'_{w}$ denotes the derivative of $\phi_{w}$ in the direction of the parameter. Accordingly, the Fisher information becomes $I(0) = \epsilon^2 \|\phi'_{w}\|^2 + \mathrm{O}(\epsilon^3)$. Hence, as $\epsilon \to 0$, the term $\|\phi'_{w}\|^2$ dominates the Fisher information and, by the Cramér–Rao lower bound \cite[Section 11.10]{cover2012elements}, governs the mean-square error of estimation. Notice the behavior of $\|\phi'_w\|^2$ under rotations, in terms of its contribution to the Fisher information and hence to the error probability, is entirely determined by the second-moment structure of $\phi_w$. This observation implies that, in the local regime, imposing spherical symmetry can be reduced to a second-moment symmetry condition on the random matrix. The sufficiency of this condition is further illustrated in Remark~\ref{rem: exact is unnecessary}.

\begin{definition}[Weak spherical symmetry]
\label{def: delta spherical symmetric random matrix}
    For random matrices $\bm{X}, \bm{Y} \in \mathbb{R}^{n \times m}$, define a second-moment distance as
    \[
    \mathrm{D}(\bm{X}, \bm{Y})= \sup_{\substack{u \in \mathbb{R}^{n}, \|u\|=1 \\ v \in \mathbb{R}^{m}, \|v\|=1}} 
    \Big| \mathbb{E} \big[ (u^{\mathrm{T}} \bm{X} v)^2 \big] - \mathbb{E} \big[(u^{\mathrm{T}} \bm{Y} v)^2 \big] \Big|.
    \]
    Given $\delta \geq 0$, a random matrix $\bm{A} \in \mathbb{R}^{n \times m}$ is \emph{$\delta$-spherically symmetric}, if for any deterministic orthogonal matrices $\bm{Q}_1 \in \mathbb{R}^{n \times n}$ and $\bm{Q}_2 \in \mathbb{R}^{m \times m}$, we have $\mathrm{D}(\bm{Q}_1^{\mathrm{T}} \bm{A} \bm{Q}_2, \bm{A}) \leq \delta$.
\end{definition}

While exact spherical symmetry trivially implies $\delta$-spherical symmetry, the latter serves as a more practical and general relaxation, as illustrated below.

\begin{example}
\label{exa: A deltass not ss}
    Given $0 < |\sigma^2-1| \leq \delta $. A random matrix $\bm{A} \in \mathbb{R}^{2 \times 2}$ has independent entries $\bm{A}_{1,1} \sim N(0, \sigma^2)$ and $\bm{A}_{1,2}, \bm{A}_{2,1}, \bm{A}_{2,2} \sim N(0, 1)$. Then, $\bm{A}$ is $\delta$-spherically symmetric, but not exactly spherically symmetric.
\end{example}

\begin{IEEEproof}
    Selecting unit vectors $u=(u_1, u_2)^{\mathrm{T}}$ and $v=(v_1, v_2)^{\mathrm{T}}$, we have $\mathbb{E}[(u^{\mathrm{T}} \bm{A} v)^2] = 1+(\sigma^2 -1) u_1^2 v_1^2$. For any orthogonal matrices $\bm{Q}_1$ and $\bm{Q}_2 \in \mathbb{R}^{2 \times 2}$, setting $\bm{Q}_1 u=(\tilde{u}_1, \tilde{u}_2)^{\mathrm{T}}$, $\bm{Q}_2 v=(\tilde{v}_1, \tilde{v}_2)^{\mathrm{T}}$, we have $\mathbb{E}[(u^{\mathrm{T}} \bm{Q}_1^{\mathrm{T}} \bm{A} \bm{Q}_2 v)^2] = 1+(\sigma^2 -1) \tilde{u}_1^2 \tilde{v}_1^2$. Then,
    \[
    \mathrm{D} \big(\bm{Q}_{1}^{\mathrm{T}} \bm{A} \bm{Q}_{2}, \bm{A} \big)  = \big| (\sigma^2 -1)(\tilde{u}_1^2 \tilde{v}_1^2 - u_1^2 v_1^2) \big| \leq |\sigma^2 -1| \leq \delta.
    \]
    which indicates $\bm{A}$ is $\delta$-spherically symmetric. The fact that $\bm{A}$ is not exactly spherically symmetric follows immediately from Lemma \ref{lem: 3 properties of ss} \ref{enu: lem_3 properties_2}.
\end{IEEEproof}

In scenarios where no prior information about $U$ and $V$ is available, it is natural to assume the attribute structures satisfy $\delta$-spherical symmetry, with a deviation bound $\delta \geq 0$ which can be chosen task-dependently; that is, the attribute structures may exhibit directional preference, but not excessively.

The main result of this paper is the following theorem, which extends the universal feature selection framework of \cite[Section 5.4]{huang2024universal} to more general scenarios, allowing directional preference in attribute structures and noise in observed data.

\begin{theorem}
\label{the: main result}
    Consider the generalized model \eqref{eq: generalized Markov chain}, in which $\hat{X}$ and $\hat{Y}$ are generated via the channel \eqref{eq: channel transfer matrices} characterized by $\eta_1$, $\eta_2$. For $k \in \{1, \dots, K_{XY}-1 \}$, select normalized feature functions $f^{k}: \mathcal{X} \to \mathbb{R}^{k}$ and $g^{k}: \mathcal{Y} \to \mathbb{R}^{k}$. Given $\delta \geq 0$, assume under each probability measures, both $\bm{\Phi}^{X \mid U}$ and $\bm{\Phi}^{Y \mid V}$ are $\delta$-spherically symmetric. Then,
    \begin{equation}
    \label{eq: main result equation}
    \begin{aligned}
    & \big( \bar{E}^{U \mid \hat{S}}(f^{k}), \bar{E}^{V \mid \hat{S}}(f^{k}), \bar{E}^{U \mid \hat{T}}(g^{k}), \bar{E}^{V \mid \hat{T}}(g^{k}) \big) \leq \\
    & \big( C_{U} \epsilon^2 k, C_{V} \epsilon^2 \sum_{i=1}^{k} \sigma_{i}^2, C_{U} \epsilon^2 \sum_{i=1}^{k} \sigma_{i}^2, C_{V} \epsilon^2 k \big) + R(\epsilon, \delta, \eta_1, \eta_2),
    \end{aligned}
    \end{equation}
    as $\epsilon \to 0$, where
    \begin{equation*}
    \small
        R(\epsilon, \delta, \eta_1, \eta_2) = \mathrm{O} \Big( \epsilon^2 \cdot \max \big\{\delta+\eta_1+\delta \eta_1, \delta+\eta_2+\delta \eta_2 \big\} \Big),
    \end{equation*}
    and $C_{U}$ and $C_{V}$ are positive constants independent of $\epsilon$, $\delta$, $\eta_1$, $\eta_2$, $k$ or $P_{\hat{X}, \hat{Y}}$, and $\sigma_{i}$ is the $i$-th largest singular value of $\tilde{\bm{B}}_{\hat{X}, \hat{Y}}$. Moreover, the equality holds when the feature vectors corresponding to $f_{i}$ and $g_{i}$ are precisely the right and left singular vectors of $\tilde{\bm{B}}_{\hat{X}, \hat{Y}}$ corresponding to $\sigma_{i}$, respectively.
\end{theorem}

\begin{remark}
\label{rem: exact is unnecessary}
    According to Theorem \ref{the: main result}, if $\delta, \eta_1, \eta_2 = \mathrm{o}(1)$ as $\epsilon \to 0$, we have $R(\epsilon, \delta, \eta_1, \eta_2) = \mathrm{o}(\epsilon^2)$. Then, \eqref{eq: main result equation} reduces to \eqref{eq: Huang Proposition 5.10 equation}, which indicates that the features selected above are asymptotically optimal and robust to the symmetry deviation $\delta$ and noise $\eta_1$, $\eta_2$. Specifically, setting $\eta_1 = \eta_2 = 0$, \eqref{eq: generalized Markov chain} boils down to \eqref{eq: Markov chain}. Therefore, Theorem \ref{the: main result} recovers \cite[Proposition 5.10]{huang2024universal} (Theorem \ref{the: Huang Proposition 5.10}), which reveals that the exact symmetry condition is unnecessary in this proposition, and a relatively small directional preference $\delta = \mathrm{o}(1)$ is tolerable. This observation further implies that the second-moment symmetry condition is sufficient for the validity of the original result.
\end{remark}

We shall require the following lemma on binary hypothesis testing, which is proved in \cite[Lemma 4.12]{huang2024universal}.

\begin{lemma}
\label{lem: Huang 4.12}
    Given a random variable $Z$ and $\epsilon > 0$, assume $W$ is an $\epsilon$-attribute of $Z$. For any $w_1, w_2 \in \mathcal{W}$, let $z^{N} = \{z_1, \dots, z_N\}$ be i.i.d. samples drawn from either $P_{Z \mid W}(\cdot \mid w_1)$ or $P_{Z \mid W}(\cdot \mid w_2)$. For normalized features $h^{k}: \mathcal{Z} \to \mathbb{R}^{k}$, define feature vectors $\psi^{Z}_{j} \leftrightarrow h_{j}, j=1,\cdots,k$, and statistics $l^{k}$ as the empirical expectation of $h^{k}$ under $z^{N}$. Let $p_{e}$ denote the error probability in deciding whether the generating distribution is $P_{Z \mid W}(\cdot \mid w_1)$ or $P_{Z \mid W}(\cdot \mid w_2)$. Then,
    \begin{equation*}
        \lim\limits_{N \to \infty} \frac{- \log p_{e}}{N} = \frac{\epsilon^2}{8} \sum_{j=1}^{k} \langle \phi^{Z \mid W}_{w_1} - \phi^{Z \mid W}_{w_2}, \psi^{Z}_{j} \rangle^2 + \mathrm{o}(\epsilon^2).
    \end{equation*}
\end{lemma}

The following lemma reveals a useful property of $\delta$-spherically symmetric distributions.

\begin{lemma}
\label{lem: delta_ss_bound}
    Given $\delta \geq 0$, let $\bm{A} \in \mathbb{R}^{n \times m}$ be a $\delta$-spherically symmetric random matrix. Then, for any deterministic matrices $\bm{G} \in \mathbb{R}^{n \times k_1}$ and $\bm{H} \in \mathbb{R}^{m \times k_2}$, there exists a constant $C$ independent of $\delta$ such that
    \begin{equation*}
        \Big| \mathbb{E}\big[\|\bm{G}^{\mathrm{T}} \bm{A} \bm{H} \|_{\mathrm{F}}^2 \big] - \frac{1}{mn} \|\bm{G}\|_{\mathrm{F}}^2  \|\bm{H}\|_{\mathrm{F}}^2 \mathbb{E} \big[\|\bm{A}\|_{\mathrm{F}}^2 \big] \Big| \leq C \delta.
    \end{equation*}
\end{lemma}

\begin{IEEEproof}
    Let $g_{i}$, $h_{j}$ be the $i$th, $j$th column of $\bm{G}$, $\bm{H}$, respectively. Then, $\|\bm{G}^{\mathrm{T}}\bm{A} \bm{H}\|_{F}^2 = \sum_{i=1}^{k_1} \sum_{j=1}^{k_2} (g_{i}^{\mathrm{T}} \bm{A} h_{j})^2$. Setting Householder matrices $\bm{Q}_{g_{i}} \in \mathbb{R}^{n \times n}$ and $\bm{Q}_{h_{j}} \in \mathbb{R}^{m \times m}$ such that $\bm{Q}_{g_{i}} g_{i} = \|g_{i} \| e_1, \ \bm{Q}_{h_{j}} h_{j} = \|h_{j} \| e_1$, we have $\mathbb{E}\big[(g_{i}^{\mathrm{T}} \bm{Q} h_{j})^2 \big] = \|g_{i}\|^2 \|h_{j}\|^2 \cdot \mathbb{E}\big[(e_{1}^{\mathrm{T}} \bm{Q}_{g_{i}}^{\mathrm{T}} \bm{A} \bm{Q}_{h_{j}} e_{1})^2 \big]$. By Definition \ref{def: delta spherical symmetric random matrix}, $\Big| \mathbb{E} \big[(e_{1}^{\mathrm{T}} \bm{Q}_{g_{i}}^{\mathrm{T}} \bm{A} \bm{Q}_{h_{j}} e_{1})^2 \big] - \mathbb{E} \big[\bm{A}_{1,1}^2 \big] \Big| \leq \delta$. Therefore,
    \begin{equation}
    \label{eq: proof_delta_ss_bound_eq4}
    \begin{split}
    & \Big| \mathbb{E}\big[\|\bm{G}^{\mathrm{T}} \bm{A} \bm{H} \|_{\mathrm{F}}^2 \big] - \|\bm{G}\|_{\mathrm{F}}^2  \|\bm{H}\|_{\mathrm{F}}^2 \mathbb{E} \big[\bm{A}_{11}^2 \big] \Big| \\ \leq &  \sum_{i=1}^{k_1} \sum_{j=1}^{k_2} \|g_{i}\|^2 \|h_{j}\|^2 \Big| \mathbb{E} \big[(e_{1}^{\mathrm{T}} \bm{Q}_{g_{i}}^{\mathrm{T}} \bm{A} \bm{Q}_{h_{j}} e_{1})^2 \big] - \mathbb{E} \big[\bm{A}_{11}^2 \big] \Big| \\ \leq & \|\bm{G}\|_{\mathrm{F}}^2  \|\bm{H}\|_{\mathrm{F}}^2 \delta.
    \end{split}
    \end{equation}
    Moreover, setting $\bm{Q}_{1}=P_{1,s}$, $\bm{Q}_{2}=P_{1,t}$ to be the permutation matrices, we have $\Big| \mathbb{E} [\bm{A}_{s,t}^2]- \mathbb{E} [\bm{A}_{1,1}^2] \Big| \leq \delta$. Hence,
    \begin{equation}
    \label{eq: proof_delta_ss_bound_eq5}
    \begin{split}
    & \Big| \frac{1}{mn} \|\bm{G}\|_{F}^2 \|\bm{H}\|_{F}^2 \mathbb{E} \big[\|\bm{A}\|_{F}^2 \big]- \|\bm{G}\|_{F}^2 \|\bm{H}\|_{F}^2 \mathbb{E} [\bm{A}_{1,1}^2] \Big| \\ \leq & \|\bm{G}\|_{F}^2 \|\bm{H}\|_{F}^2 \Big| \frac{1}{mn} \sum_{s=1}^{n} \sum_{t=1}^{m} \Big( \mathbb{E} [\bm{A}_{s,t}^2]- \mathbb{E} [\bm{A}_{1,1}^2] \Big) \Big| \\ \leq & \|\bm{G}\|_{F}^2 \|\bm{H}\|_{F}^2 \delta.
    \end{split}
    \end{equation}
    Combining \eqref{eq: proof_delta_ss_bound_eq4}, \eqref{eq: proof_delta_ss_bound_eq5} and the triangle inequality, we have
    \[
    \Big| \mathbb{E}\big[\|\bm{G}^{\mathrm{T}} \bm{A} \bm{H} \|_{\mathrm{F}}^2 \big]-\frac{1}{mn} \|\bm{G}\|_{\mathrm{F}}^2  \|\bm{H}\|_{\mathrm{F}}^2 \mathbb{E} \big[\|\bm{A}\|_{\mathrm{F}}^2 \big] \Big| \leq C \delta,
    \]
    where $C = 2 \|\bm{G}\|_{F}^2 \|\bm{H}\|_{F}^2$ is independent of $\delta$.
\end{IEEEproof}

The following lemma demonstrates that the weakly spherically symmetric properties remain preserved under left multiplication by a deterministic matrix.

\begin{lemma}
\label{lem: A deltass BA gammass}
    Given $\delta \geq 0$, let $\bm{A} \in \mathbb{R}^{n \times m}$ be a $\delta$-spherically symmetric random matrix, and $\bm{B} \in \mathbb{R}^{n \times n}$ be a deterministic matrix with singular values $\sigma_1 \geq \cdots \geq \sigma_{n} \geq 0$. Then, $\bm{BA}$ is $\gamma$-spherically symmetric, where $\gamma = (\alpha + \delta)(\sigma_1^2 - \sigma_{n}^2) + \sigma_1^2 \delta$ with $\alpha = \frac{\mathbb{E}[\|\bm{A}\|_{\mathrm{F}}^2]}{mn}$.
\end{lemma}

\begin{IEEEproof}
    For unit vectors \(x \in \mathbb{R}^{n}\) and \(y \in \mathbb{R}^{m}\), define
    \[
    f(x,y) = \mathbb{E}\big[(x^{\mathrm{T}} \bm{A} y)^2\big],
    \qquad
    \alpha = \frac{\mathbb{E}[\|\bm{A}\|_{\mathrm{F}}^2]}{mn}.
    \]
    Since \(\bm{A}\) is \(\delta\)-spherically symmetric, for any unit vectors \(x,y,x',y'\), we have $\big|f(x,y)-f(x',y')\big| \leq \delta$. Hence all values of \(f\) lie in an interval of length at most \(\delta\). Moreover,
    \[
    \alpha = \mathbb{E}_{U,V}\big[f(U,V)\big],
    \]
    where \(U\) and \(V\) are independent uniform unit vectors, because
    \(\mathbb{E}[UU^{\mathrm{T}}] = I_n/n\) and
    \(\mathbb{E}[VV^{\mathrm{T}}] = I_m/m\). Therefore, \(\alpha\) lies in the same interval, and $\big|f(x,y)-\alpha\big| \leq \delta$.

    For \(\bm{B}^{\mathrm{T}}u \neq 0\), let
    \[
    \tilde{u} = \frac{\bm{B}^{\mathrm{T}}u}{\|\bm{B}^{\mathrm{T}}u\|}.
    \]
    Then
    \begin{equation}
    \label{eq: BA deltass proof_2}
        \mathbb{E}\big[(u^{\mathrm{T}} \bm{B}\bm{A} v)^2\big]
        =
        u^{\mathrm{T}} \bm{B} \bm{B}^{\mathrm{T}} u \cdot f(\tilde{u},v).
    \end{equation}
    If \(\bm{B}^{\mathrm{T}}u = 0\), the original quadratic form vanishes; this case can be handled separately or by continuity and does not affect the final bound.

    For any orthogonal matrices \(\bm{Q}_1 \in \mathbb{R}^{n \times n}\) and
    \(\bm{Q}_2 \in \mathbb{R}^{m \times m}\), let $\hat{u} = \bm{Q}_1 u$, $\hat{v} = \bm{Q}_2 v$. If \(\bm{B}^{\mathrm{T}} \hat{u} \neq 0\), let
    \[
    \tilde{\hat{u}} = \frac{\bm{B}^{\mathrm{T}}\hat{u}}{\|\bm{B}^{\mathrm{T}}\hat{u}\|}.
    \]
    Then,
    \[
    \mathbb{E}\big[(\hat{u}^{\mathrm{T}} \bm{B}\bm{A} \hat{v})^2\big]
    =
    \hat{u}^{\mathrm{T}} \bm{B}\bm{B}^{\mathrm{T}} \hat{u} \cdot f(\tilde{\hat{u}},\hat{v})
    =
    u^{\mathrm{T}} \bm{Q}_1^{\mathrm{T}} \bm{M} \bm{Q}_1 u \cdot f(\tilde{\hat{u}},\hat{v}),
    \]
    where \(\bm{M} = \bm{B}\bm{B}^{\mathrm{T}}\).

    Consider $\Delta = \big|
    \mathbb{E}[(\hat{u}^{\mathrm{T}} \bm{B}\bm{A} \hat{v})^2] - \mathbb{E}[(u^{\mathrm{T}} \bm{B}\bm{A} v)^2]
    \big|$. Substituting \eqref{eq: BA deltass proof_2}, we obtain
    \[
    \Delta
    =
    \left|
    (u^{\mathrm{T}}\bm{Q}_1^{\mathrm{T}} \bm{M} \bm{Q}_1 u) \cdot f(\tilde{\hat{u}},\hat{v})
    -
    (u^{\mathrm{T}} \bm{M} u) \cdot f(\tilde{u},v)
    \right|.
    \]
    By the triangle inequality,
    \begin{align*}
        \Delta
        \leq{}&
        (u^{\mathrm{T}}\bm{Q}_1^{\mathrm{T}}\bm{M} \bm{Q}_1 u)
        \cdot
        \big| f(\tilde{\hat{u}},\hat{v}) - f(\tilde{u},v) \big| \\
        &+
        \big| f(\tilde{u},v) \big|
        \cdot
        \big| u^{\mathrm{T}} \bm{Q}_1^{\mathrm{T}} \bm{M} \bm{Q}_1 u - u^{\mathrm{T}} \bm{M} u \big|.
    \end{align*}

    As demonstrated, $\big| f(\tilde{\hat{u}},\hat{v}) - f(\tilde{u},v) \big| \leq \delta$, $\big| f(\tilde{u},v) \big| \leq \alpha + \delta$. Since \(\bm{M} = \bm{B}\bm{B}^{\mathrm{T}}\), we have $\lambda_{\max}(\bm{M}) = \sigma_1^2$, $\lambda_{\min}(\bm{M}) = \sigma_n^2$. Therefore,
    \[
    u^{\mathrm{T}} \bm{Q}_1^{\mathrm{T}} \bm{M} \bm{Q}_1 u \leq \sigma_1^2,
    \]
    and
    \[
    \big| u^{\mathrm{T}} \bm{Q}_1^{\mathrm{T}} \bm{M} \bm{Q}_1 u - u^{\mathrm{T}} \bm{M} u \big|
    \leq
    \sigma_1^2 - \sigma_n^2.
    \]
    Combining these results gives
    \[
    \Delta
    \leq
    \sigma_1^2 \delta
    +
    (\alpha + \delta)(\sigma_1^2 - \sigma_n^2)
    =
    (\alpha + \delta)(\sigma_1^2 - \sigma_n^2) + \sigma_1^2 \delta
    =
    \gamma.
    \]
    Taking the supremum over all unit vectors \(u,v\) and all orthogonal matrices
    \(\bm{Q}_1,\bm{Q}_2\), we obtain
    \[
    \mathrm{D}(\bm{Q}_1^{\mathrm{T}} \bm{B}\bm{A} \bm{Q}_2, \bm{B}\bm{A}) \leq \gamma.
    \]
    Thus, \(\bm{B}\bm{A}\) is \(\gamma\)-spherically symmetric.
\end{IEEEproof}

\begin{remark}
    Lemma \ref{lem: A deltass BA gammass} provides an upper bound on the symmetry deviation of $\bm{BA}$. To verify its validity, first suppose $\bm{B} = \bm{I}_{n}$, then $\bm{BA} \stackrel{\text{d}}{=} \bm{A}$, and by Lemma \ref{lem: A deltass BA gammass}, we obtain $\gamma = \delta$. Next, assume $\bm{A}$ is exactly spherically symmetric, i.e., $\delta = 0$. In this case, Lemma \ref{lem: A deltass BA gammass} gives $\gamma = (\sigma_1^2 - \sigma_{n}^2) \alpha$, showing that the symmetry deviation of $\bm{BA}$ arises solely from the range among the singular values of $\bm{B}$.
\end{remark}

Applying Lemma \ref{lem: A deltass BA gammass} yields the following lemma.

\begin{lemma}
\label{lem: remain gammass}
    Given $\delta \geq 0$, assume both $\bm{\Phi}^{X \mid U}$ and $\bm{\Phi}^{Y \mid V}$ are $\delta$-spherically symmetric. Suppose $\hat{X}$ and $\hat{Y}$ are generated via the channel \eqref{eq: channel transfer matrices} characterized by $\eta_1$, $\eta_2$. Then, $\bm{\Phi}^{\hat{X} \mid U}$ is $\gamma_1$-spherically symmetric, and $\bm{\Phi}^{\hat{Y} \mid V}$ is $\gamma_2$-spherically symmetric, where $\gamma_1 = \mathrm{O}(\delta + \eta_1 + \delta \eta_1)$, $\gamma_2 = \mathrm{O}(\delta + \eta_2 + \delta \eta_2)$.
\end{lemma}

\begin{IEEEproof}
    Define the matrix $\bm{B}_{X,\hat{X}}$ such that
    \begin{equation*}
        (\bm{B}_{X,\hat{X}})_{j,i} = \frac{P_{X,\hat{X}}(x_{i}, \hat{x}_{j})}{\sqrt{P_{X}(x_{i})} \sqrt{P_{\hat{X}}(\hat{x}_{j})}}.
    \end{equation*}
    Since $U \leftrightarrow X \leftrightarrow \hat{X}$ forms a Markov chain, for any $u \in \mathcal{U}$, we have $P_{\hat{X} \mid U}(\hat{x} \mid u) = \sum_{x} P_{\hat{X} \mid X}(\hat{x} \mid x) P_{X \mid U}(x \mid u)$, and $P_{\hat{X}}(\hat{x}) = \sum_{x} P_{\hat{X} \mid X}(\hat{x} \mid x) P_{X}(x)$. Accordingly, $\phi^{\hat{X} \mid U}_{u} = \bm{B}_{X, \hat{X}} \cdot \phi^{X \mid U}_{u}$. Thus, $\Phi^{\hat{X} \mid U} = \bm{B}_{X, \hat{X}} \cdot \Phi^{X \mid U}$.
    
    Let $\bm{D}_{X}$, $\bm{D}_{\hat{X}}$ be the diagonal matrices with entries $P_{X}(x)$, $P_{\hat{X}}(\hat{x})$, respectively. Since $\bm{B}_{X, \hat{X}} = \bm{D}_{\hat{X}}^{-1/2} \bm{P}_{\hat{X} \mid X} \bm{D}_{X}^{1/2}$ and $\bm{B}_{X, \hat{X}}^{\mathrm{T}} = \bm{D}_{X}^{-1/2} \bm{P}_{X \mid \hat{X}} \bm{D}_{\hat{X}}^{1/2}$, the matrix $\bm{B}_{X, \hat{X}} \bm{B}_{X, \hat{X}}^{\mathrm{T}}$ is similar to $\bm{P}_{\hat{X} \mid X} \bm{P}_{X \mid \hat{X}}$, thus has the same eigenvalues.
    
    According to \eqref{eq: channel transfer matrices}, we have $\bm{D}_{\hat{X}} = \bm{D}_{X} + \mathrm{O}(\bm{\eta}_1)$. Then, the reverse channel matrix $\bm{P}_{X \mid \hat{X}} = \bm{D}_{X} \bm{P}_{\hat{X} \mid X}^{\mathrm{T}} \bm{D}_{\hat{X}}^{-1} = \bm{I} + \mathrm{O}(\bm{\eta}_1)$. Thus, $\bm{P}_{\hat{X} \mid X} \bm{P}_{X \mid \hat{X}} = \bm{I} + \mathrm{O}(\bm{\eta}_1)$. Let $\bm{P}' = \bm{P}_{\hat{X} \mid X} \bm{P}_{X \mid \hat{X}}$. Since both $\bm{P}_{\hat{X} \mid X}$ and $\bm{P}_{X \mid \hat{X}}$ are column-stochastic matrices, their product $\bm{P}'$ is still column-stochastic. Since $\rho(\bm{P}') \leq \|\bm{P}'\|_{1} = 1$ and $
    (\bm{P}')^{\mathrm{T}} \bm{1} = \bm{1}$, we have $\lambda_{\text{max}} (\bm{P}') = 1$. In addition, $\bm{P}' = \bm{I} + \mathrm{O}(\bm{\eta}_1)$ gives $\lambda_{\text{min}} (\bm{P}') = 1 - \mathrm{O}(\eta_1)$. Then, $\sigma_{\text{max}}^2 (\bm{B}_{X, \hat{X}}) - \sigma_{\text{min}}^2 (\bm{B}_{X, \hat{X}}) = \mathrm{O}(\eta_1)$. By Lemma \ref{lem: A deltass BA gammass}, $\bm{\Phi}^{\hat{X} \mid U}$ is $\gamma_1$-spherically symmetric where $\gamma_1 = \mathrm{O}(\delta + \eta_1 + \delta \eta_1)$.

    The proof for $\bm{\Phi}^{\hat{Y} \mid V}$ is similar.
\end{IEEEproof}

Markov chain $U \leftrightarrow X \leftrightarrow Y \leftrightarrow V$ implies $\bm{\Phi}^{Y \mid U} = \tilde{\bm{B}}_{X,Y} \bm{\Phi}^{X \mid U}$ and $\bm{\Phi}^{X \mid V} = \tilde{\bm{B}}_{X,Y}^{\mathrm{T}} \bm{\Phi}^{Y \mid V}$. The following lemma reveals analogous results for the generalized model \eqref{eq: generalized Markov chain}.

\begin{lemma}
\label{lem: Phi = B Phi}
    In the generalized model \eqref{eq: generalized Markov chain}, suppose $\hat{X}$, $\hat{Y}$ are generated via the channel \eqref{eq: channel transfer matrices} characterized by $\eta_1$, $\eta_2$. Then,
    \begin{equation*}
        \bm{\Phi}^{\hat{Y} \mid U} = \tilde{\bm{B}}_{\hat{X}, \hat{Y}} \bm{\Phi}^{\hat{X} \mid U} + \mathrm{O}(\bm{\eta}_1), \bm{\Phi}^{\hat{X} \mid V} = \tilde{\bm{B}}_{\hat{X}, \hat{Y}}^{\mathrm{T}} \bm{\Phi}^{\hat{Y} \mid V} + \mathrm{O}(\bm{\eta}_2).
    \end{equation*}
\end{lemma}

\begin{IEEEproof}
    The channel \eqref{eq: channel transfer matrices} gives $P_{\hat{X} \mid X}(\hat{x} \mid x) = 1 + \mathrm{O}(\eta_1)$. Since $U \leftrightarrow X \leftrightarrow \hat{Y}$ forms a Markov chain, for any $u \in \mathcal{U}$,
    \begin{align*}
        P_{\hat{Y} \mid U}(\hat{y} \mid u) & = \sum_{x}  P_{\hat{Y} \mid X}(\hat{y} \mid x) P_{X \mid U}(x \mid u) \\ & = \sum_{\hat{x}}  P_{\hat{Y} \mid \hat{X}}(\hat{y} \mid \hat{x}) P_{\hat{X} \mid U}(\hat{x} \mid u) + \mathrm{O}(\eta_1).
    \end{align*}
    Then, the information vector $\phi^{\hat{Y} \mid U}_{u} = \tilde{\bm{B}}_{\hat{X}, \hat{Y}} \phi^{\hat{X} \mid U}_{u} + \mathrm{O}(\eta_1)$. Hence, $\bm{\Phi}^{\hat{Y} \mid U} = \tilde{\bm{B}}_{\hat{X}, \hat{Y}} \bm{\Phi}^{\hat{X} \mid U} + \mathrm{O}(\bm{\eta}_1)$.
    
    The second equality can be verified in the same way.
\end{IEEEproof}

\begin{IEEEproof}[Proof of Theorem \ref{the: main result}]
    The proof largely follows the lines of the proof of \cite[Proposition 5.10]{huang2024universal}, with the incorporation of Lemma \ref{lem: Huang 4.12} and modifications based on Lemmas \ref{lem: delta_ss_bound}, \ref{lem: remain gammass}, and \ref{lem: Phi = B Phi}.
    
    For $p_{e}^{U|\hat{T}}(g^{k},\mathcal{C}^{\epsilon}_{\mathcal{X}}(P_{X}))$, setting $\mathcal{C}^{\epsilon}_{\mathcal{X}}(P_{X}) = \mathcal{C}$, we have
    \begin{align}
        & \lim\limits_{N \to \infty} - \frac{\log p_{e}^{U|\hat{T}}(g^{k},  \mathcal{C})}{N} = \lim\limits_{N \to \infty} - \frac{\log p_{e}^{U|\hat{T}}(g^{k}, u_{\mathcal{C}}^{*}, u_{\mathcal{C}}^{*'})}{N} \label{eq: pairwise e e 1} \\ = & \frac{\epsilon^2}{8} \sum_{j=1}^{k} \langle \phi^{\hat{Y} \mid U}_{u_{\mathcal{C}}^{*}} - \phi^{\hat{Y} \mid U}_{u_{\mathcal{C}}^{*'}}, \psi_{j}^{\hat{Y}} \rangle^2 + \mathrm{o}(\epsilon^2), \label{eq: pairwise e e 2}
    \end{align}
    where 1) to obtain \eqref{eq: pairwise e e 1}, we have used $p_{e}^{U\mid\hat{T}}(g^{k}, u, u')$ to represent the pairwise error probability for distinguishing $u$ and $u'$ based on $\hat{T}^{k}$, and defined $(u_{\mathcal{C}}^{*}, u_{\mathcal{C}}^{*'}) = \arg\max_{(u_{\mathcal{C}}, u_{\mathcal{C}}')} p_{e}^{U \mid \hat{T}}(g^{k}, u_{\mathcal{C}}, u_{\mathcal{C}}')$ as the least distinguishable hypothesis pair in $\mathcal{C}$; 2) to obtain \eqref{eq: pairwise e e 2}, we have applied Lemma \ref{lem: Huang 4.12}, with $\psi_{j}^{\hat{Y}}$ denoting the corresponding feature vector of $g_{j}, j=1 \cdots, k$. Thus, the error exponent in decisions for $U$ based on $\hat{T}^{k}$ is derived by
    \begin{align}
        & \ \bar{E}^{U \mid \hat{T}}(g^{k})  = \mathbb{E}_{\mu_{U}} \Big[\lim\limits_{N \to \infty} - \frac{\log p_{e}^{U \mid \hat{T}}(g^{k}, \mathcal{C}^{\epsilon}_{\mathcal{X}}(P_{X}))}{N} \Big] \label{eq: proof_prop_noisy_0} \\ = & \ \frac{\epsilon^2}{8} \mathbb{E}_{\mu_{U}} \Big[ \big\|(\bm{\Psi}^{\hat{Y}})^{\mathrm{T}} \Phi^{\hat{Y} \mid U} (\bm{e}_{u^{*}} - \bm{e}_{u^{*'}}) \big\|_{\mathrm{F}}^2 \Big] + \mathrm{o}(\epsilon^2) \label{eq: proof_prop_noisy_2} \\ = & \ \frac{\epsilon^2}{8} \mathbb{E}_{\mu_{U}} \Big[ \big\|(\bm{\Psi}^{\hat{Y}})^{\mathrm{T}} \tilde{\bm{B}}_{\hat{X}, \hat{Y}} \bm{\Phi}^{\hat{X} \mid U} (\bm{e}_{u^{*}} - \bm{e}_{u^{*'}}) \big\|_{\mathrm{F}}^2 \Big] + \mathrm{O}(\eta_1 \epsilon^2) \label{eq: proof_prop_noisy_3} \\ = & \ \frac{\epsilon^2 \mathbb{E}_{\mu_{U}} \Big[\big\| \Phi^{\hat{X} \mid U} \big\|_{\mathrm{F}}^2 \Big]}{4 |\mathcal{X}| |\mathcal{U}|} \big\|(\bm{\Psi}^{\hat{Y}})^{\mathrm{T}} \tilde{\bm{B}}_{\hat{X}, \hat{Y}} \big\|_{\mathrm{F}}^2 + r(\epsilon, \delta, \eta_1) \label{eq: proof_prop_noisy_4} \\ \leq & \ C_{U} \epsilon^2 \sum_{i=1}^{k} \sigma_{i}^2 + R(\epsilon, \delta, \eta_1, \eta_2), \label{eq: proof_prop_noisy_5}
    \end{align}
    where 1) to obtain \eqref{eq: proof_prop_noisy_2} we have invoked \eqref{eq: pairwise e e 2}, with $(u^{*}, u^{*'})$ denoting the least distinguishable hypothesis pair in each configuration $\mathcal{C}^{\epsilon}_{\mathcal{X}}(P_{X})$, and $\bm{\Psi}^{\hat{Y}} = (\psi^{\hat{Y}}_{1}, \cdots, \psi^{\hat{Y}}_{k})$; 2) to obtain \eqref{eq: proof_prop_noisy_3} we have employed Lemma \ref{lem: Phi = B Phi}; 3) to obtain \eqref{eq: proof_prop_noisy_4} we have used Lemma \ref{lem: delta_ss_bound} and Lemma \ref{lem: remain gammass}, with $r(\epsilon, \delta, \eta_1) = \mathrm{O} \big( (\delta+\eta_1+\delta \eta_1)\epsilon^2 \big)$; 4) to obtain \eqref{eq: proof_prop_noisy_5} we have defined the constant $C_{U} = \frac{\mathbb{E}_{\mu_{U}} \big[\| \Phi^{\hat{X} \mid U} \|_{\mathrm{F}}^2 \big]}{4 |\mathcal{X}| |\mathcal{U}|}$, and noted that the upper bound of $\big\|(\bm{\Psi}^{\hat{Y}})^{\mathrm{T}} \tilde{\bm{B}}_{\hat{X}, \hat{Y}} \big\|_{\mathrm{F}}^2$ is attained by choosing $\bm{\Psi}^{\hat{Y}}$ as the corresponding left singular vectors of $\tilde{\bm{B}}_{\hat{X}, \hat{Y}}$ \cite[Corollary 4.3.39]{horn2012matrix}.

    The error exponent in decisions for $U$ based on $\hat{S}^{k}$ can be derived in a similar way:
    \begin{align}
        & \ \bar{E}^{U \mid \hat{S}}(f^{k})  = \mathbb{E}_{\mu_{U}} \Big[\lim\limits_{N \to \infty} - \frac{\log p_{e}^{U \mid \hat{S}}(f^{k}, \mathcal{C}^{\epsilon}_{\mathcal{X}}(P_{X}))}{N} \Big] \label{eq: proof_prop_noisy_11} \\ = & \ \frac{\epsilon^2}{8} \mathbb{E}_{\mu_{U}} \Big[ \big\|(\bm{\Psi}^{\hat{X}})^{\mathrm{T}} \Phi^{\hat{X} \mid U} (\bm{e}_{u_{*}} - \bm{e}_{u_{*}'}) \big\|_{\mathrm{F}}^2 \Big] + \mathrm{o}(\epsilon^2) \label{eq: proof_prop_noisy_12} \\ = & \ \frac{\epsilon^2 \mathbb{E}_{\mu_{U}} \Big[\big\| \Phi^{\hat{X} \mid U} \big\|_{\mathrm{F}}^2 \Big]}{4 |\mathcal{X}| |\mathcal{U}|} \big\| \bm{\Psi}^{\hat{X}} \big\|_{\mathrm{F}}^2 + r(\epsilon, \delta, \eta_1) \label{eq: proof_prop_noisy_14} \\ \leq & \ C_{U} \epsilon^2 k + R(\epsilon, \delta, \eta_1, \eta_2) \label{eq: proof_prop_noisy_15}
    \end{align}
    where to obtain \eqref{eq: proof_prop_noisy_14} we have used Lemma \ref{lem: delta_ss_bound} and Lemma \ref{lem: remain gammass}.

    The proofs for $\bar{E}^{V \mid \hat{S}}(f^{k})$ and $\bar{E}^{V \mid \hat{T}}(g^{k})$ are symmetric to $\bar{E}^{U \mid \hat{T}}(g^{k})$ and $\bar{E}^{U \mid \hat{S}}(f^{k})$, respectively.

    The proof is then completed.
\end{IEEEproof}

\section{Discussion}

Considering the generalized model \eqref{eq: generalized Markov chain}, we construct features solely from the noisy observations $\hat{X}$ and $\hat{Y}$ to perform inference tasks concerning the latent attributes $U$ of $X$ and $V$ of $Y$. We formally introduce the notion of $\delta$-spherical symmetry (Definition~\ref{def: delta spherical symmetric random matrix}) and elucidate that second-moment conditions play a central role in the proposed problem. In the absence of prior knowledge about the attributes, it is natural to assume that the underlying information structures satisfy $\delta$-spherical symmetry, which permits the second moment to exhibit a certain degree of directional preference, yet not excessively so. Under this relaxed condition, we develop a universal feature selection framework that extends \cite[Proposition 5.10]{huang2024universal} to more general scenarios (Theorem \ref{the: main result}). We further demonstrate that our result recovers theirs as a special case when $\eta_1 = \eta_2 = 0$ and $\delta = o(1)$ as $\epsilon \to 0$ (Remark \ref{rem: exact is unnecessary}), thereby revealing that exact rotational invariance is not necessary for the original proposition; a relatively small directional preference is tolerable.

Generally, if $\delta, \eta_1, \eta_2 = o(1)$ as $\epsilon \to 0$, the residual term in \eqref{eq: main result equation} satisfies $R(\epsilon, \delta, \eta_1, \eta_2) = o(\epsilon^2)$. Accordingly, the features obtained via the singular value decomposition of $\tilde{\bm{B}}_{\hat{X}, \hat{Y}}$ achieve asymptotic optimality in the local regime. For cases where the parameters $\delta, \eta_1, \eta_2$ are not relatively small (e.g., are of order $\mathrm{O}(1)$), the proposed universal feature selection framework remains meaningful: the error bound is then governed by the residual term $R(\epsilon, \delta, \eta_1, \eta_2)$, which depends explicitly on $\delta$, $\eta_1$, and $\eta_2$. This demonstrates that the framework is robust against both symmetry deviation $\delta$ and observation noise $\eta_1, \eta_2$, as the error bound remains well-controlled even when these parameters are not vanishingly small. Moreover, the framework requires no prior knowledge of the latent attributes, rendering it applicable to a variety of domains where the attribute distributions are unknown or undetermined a priori. In summary, our framework extends to a broad class of approximately symmetric and noisy environments across a diverse range of inference tasks. We believe this work provides a robust, theoretically grounded tool for universal feature selection in practical, imperfect data settings within machine learning and statistical inference.

\balance

\end{document}